\documentclass[10pt,english,conference]{IEEEtran}
\usepackage[T1]{fontenc}
\usepackage[utf8]{inputenc}
\usepackage{amsthm}
\usepackage{amsmath}
\usepackage{graphicx}
\usepackage{amssymb}

\makeatletter
  \theoremstyle{definition}
  \newtheorem{defn}{Definition}
  \theoremstyle{plain}
  \newtheorem{lem}{Lemma}
\theoremstyle{plain}
\newtheorem{thm}{Theorem}
  \theoremstyle{plain}
  \newtheorem{cor}{Corollary}
  \theoremstyle{remark}
  \newtheorem{rem}{Remark}
 \theoremstyle{definition}
  \newtheorem{example}{Example}
  \theoremstyle{plain}
  \newtheorem*{conjecture*}{Conjecture}

\IEEEoverridecommandlockouts
\usepackage{cite}

\ifCLASSINFOpdf
\else
\fi
\usepackage{indentfirst}

\@ifundefined{showcaptionsetup}{}{%
 \PassOptionsToPackage{caption=false}{subfig}}
\usepackage{subfig}
\makeatother

\usepackage{babel}

\begin{document}

\title{Towards the Capacity Region of Multiplicative Linear Operator Broadcast
Channels }

\author{\IEEEauthorblockN{Yimin Pang}\IEEEauthorblockA{Department of
Information Science\\
 and Electronic Engineering\\
 Zhejiang University\\
 Hangzhou, China, 310027\\
 Email: yimin.pang@zju.edu.cn } \and\IEEEauthorblockN{Thomas
Honold\IEEEauthorrefmark{1}} \IEEEauthorblockA{Department of
Information Science\\
 and Electronic Engineering\\
 Zhejiang University\\
 Hangzhou, China, 310027\\
 Email: honold@zju.edu.cn} %
\thanks{Supported by the National Natural Science Foundation of China under
Grant No. 60872063%
}}
\maketitle
\begin{abstract}
Recent research indicates that packet transmission employing random
linear network coding can be regarded as transmitting subspaces over
a linear operator channel (LOC). In this paper we propose the framework
of linear operator broadcast channels (LOBCs) to model packet broadcasting
over LOCs, and we do initial work on the capacity region of constant-dimension
multiplicative LOBCs (CMLOBCs), a generalization of broadcast erasure
channels. Two fundamental problems regarding CMLOBCs are addressed--finding
necessary and sufficient conditions for degradation and deciding whether
time sharing suffices to achieve the boundary of the capacity region
in the degraded case.\end{abstract}

\section{Introduction}

Random linear network coding \cite{1705002} is an efficient alternative
to achieve the network capacity proposed in \cite{850663}. In a random
linear network coding channel packets are transmitted in generations
and are regarded as $n$-dimensional row vectors over some finite
field $\mathbb{F}_{q}$. Due to the subspace preserving property,
packet transmission over an acyclic noisy network may be thought of
as conveying subspaces over a \emph{linear operator channel} (LOC)
\cite{4567581}, whose input and output symbols are taken from the
set of all subspaces of $\mathbb{F}_{q}^{m}$ (referred to as {}``ambient
space''). In \cite{5429134} Silva et al.\ investigated the capacity
of a random linear network coding channel with matrices as input/output
symbols. Later, by regarding a LOC as a particular DMC, Uch{ô}a-Filho
and N{ó}brega \cite{uchoa2010capacity} studied the capacity of
constant dimension multiplicative LOCs. Yang et al. \cite{yang2010optimality,yang2010linear}
considered general non-constant multiplicative LOC capacity. In \cite{jafari2009capacity}
the rate region of multiple source access LOCs was investigated.

We will denote the set of all $i$-dimensional subspaces of $\mathbb{F}_{q}^{m}$
by $\mathcal{P}(\mathbb{F}_{q}^{m},i)$. The following notation will
be used in the sequel. Symbols $\mathsf{X}$, $\mathsf{Y}$ and $\mathsf{U}$
denote random variables with values from subspace alphabets $\mathfrak{X}$,
$\mathfrak{Y}$, respectively $\mathfrak{U}$. The symbols $X$, $Y$
and $U$ denote subspaces in $\mathfrak{X}$, $\mathfrak{Y}$ and
$\mathfrak{U}$, respectively.

\textit{Constant-dimension multiplicative LOCs (CMLOCs)} deserve our
interest, since they capture most packet transmission scenarios. A
precise definition of CMLOCs from the information theory point-of-view
is the following. 
\begin{defn}
\label{dfn:CMLOC} A \emph{constant-dimension multiplicative LOC}
(CMLOC) of constant dimension $l$ is a discrete memoryless channel
(DMC) with input alphabet $\mathfrak{X}=\mathcal{P}(\mathbb{F}_{q}^{m},l)$
, output alphabet $\mathfrak{Y}=\bigcup_{i=0}^{l}\mathcal{P}(\mathbb{F}_{q}^{m},i)$
and transfer probabilities $p(Y|X)=p_{\mathsf{Y}|\mathsf{X}}(Y|X)=p(\mathsf{Y}=Y|\mathsf{X}=X)$
($X\in\mathfrak{X}$, $Y\in\mathfrak{Y}$) satisfying \begin{equation}
p(Y|X)=\begin{cases}
\frac{\epsilon_{\text{dim}(Y)}}{\binom{l}{\text{dim}(Y)}_{q}} & \text{if}\,\, Y\subseteq X,\\
0 & \text{otherwise}.\end{cases}\label{eq:LOCTransition}\end{equation}
 Here $\epsilon_{i}$, $0\leq i\leq l$, denotes the probability of
receiving an $i$-dimensional subspace, and $\binom{l}{i}_{q}$ is
the familiar Gaussian binomial coefficient. 
\end{defn}
Our definition of a CMLOC is slightly different from that in \cite{uchoa2010capacity},
where instead of $\boldsymbol{\mathbf{\epsilon}}$ the rank deficiency
distribution $p_{\rho}(i)_{0\leq i\leq l}$ (related to our distribution
by $p_{\rho}(i)=\epsilon_{l-i}$) occurs. In our case the total erasure
probability is $\epsilon_{0}+\epsilon_{1}+\dots+\epsilon_{l-1}=1-\epsilon_{l}$,
and $\epsilon_{l}$ is the probability of error-free transmission.
The capacity of a CMLOC is given in \cite[Th.~4]{uchoa2010capacity}.

As we know, only packet multicasting benefits from network coding
and on the other hand multicasting at a constant rate would either
starve receivers with high band-width or overwhelm those with a poor
connection. This provides our motivation to investigate broadcasting
over LOCs.

Basic knowledge on broadcast channels can be found in \cite{cover1991elements,1055652,720547}.
Recent work showed that the computation of the capacity region of
a discrete memoryless degraded broadcast channel is a non-convex DC
problem \cite{4595282}. Later Yasui et al.\cite{5513525} applied
the Arimoto-Blahut algorithm \cite{arimoto72,blahut72} for numerically
computing the channel capacity.

The framework of general Linear Operator Broadcast Channels (LOBCs)
is presented in Section II with emphasis on constant-dimension multiplicative
LOBCs (CMLOBCs), a generalization of the well-known binary erasure
broadcast channel (BEBC). Two fundamental questions about CMLOBCs
are addressed: First, when will a CMLOBC be stochastically degraded?
While for BEBCs the solution is quite obvious, for CMLOBCs the rich
structure of possible erasures makes the problem quite intriguing.
Our solution is discussed in Section III. Second, in the case of a
degraded CMLOBC is time sharing sufficient to exhaust the capacity
region?---for BEBCs the answer is {}``yes'' and is again fairly
obvious \cite{1523761}. In Section IV, we prove that for CMLOBCs
this is not always true and further discuss the shape of the capacity
region of CMLOBCs with subspaces taken from the projective plain $\operatorname{PG}(2,2)$.
Plenty of numerical analysis are shown on different cases of CMLOBCs
over $\operatorname{PG}(2,2)$ , via Arimoto-Blahut type algorithm
in \cite{5513525}. Section V concludes the paper. Proofs can be found
in the appendix (Section VII).

\section{Linear Operator Broadcast Channels (LOBCs)}

\subsection{LOBC Module}

We consider the case of a multiple user LOC where a sender communicates
with $K$ receivers $u_{1}$, $u_{2}$,...,$u_{K}$ simultaneously.
The subchannels from the sender to $u_{k}$, $k=1,2,...,K$, are linear
operator channels with input and output alphabets $\mathfrak{X},\mathfrak{Y}\subseteq\bigcup_{i=0}^{m}\mathcal{P}(\mathbb{F}_{q}^{m},i)$,
where $m$ and $q$ are fixed. Let $\mathsf{X},\mathsf{Y}_{1},\dots,\mathsf{Y}_{k}$
be the corresponding random variables. The output at every receiver
is taken 
subject to some joint transfer probability distribution $p(Y_{1},Y_{2},\dots,Y_{k}|X)=p_{\mathsf{Y_{1}},\mathsf{Y}_{2},\dots,\mathsf{Y}_{k}|\mathsf{X}}(Y_{1},Y_{2},\dots,Y_{k}|X)=p(\mathsf{Y}_{1}=Y_{1},\mathsf{Y}_{2}=Y_{2},\dots,\mathsf{Y}_{k}=Y_{k}|\mathsf{X}=X)$.
Such a channel is called \emph{Linear Operator Broadcast Channel (LOBC)}.
For simplicity we restrict ourselves to a LOBC with two receivers
and let $\mathfrak{M}_{1}$, $\mathfrak{M}_{2}$ be the alphabets
of private messages for user $u_{1}$ and $u_{2}$, respectively. 
\begin{defn}
A \emph{broadcast (multishot) subspace code} of length $n$ for the
LOBC consists of a set $\mathfrak{C}\subseteq\mathfrak{X}^{n}$ of
codewords and a corresponding encoder/decoder pair. The LOBC encoder
$\gamma\colon\mathfrak{M}_{1}\times\mathfrak{M}_{2}\to\mathfrak{C}$
maps a message pair $(M_{1},M_{2})$ to a codeword $\mathbf{X}=(X_{1},\dots,X_{n})\in\mathfrak{C}$
(for every transmission generation). The LOBC decoder $\delta=(\delta_{1},\delta_{2})$
consists of two decoding functions $\delta_{i}\colon\mathfrak{Y}^{n}\to\mathfrak{M}_{i}$
($i=1,2$) and maps the corresponding pair $(Y_{1},Y_{2})\in\mathfrak{Y}^{n}\times\mathfrak{Y}^{n}$
of received words to the message pair $(\hat{M}_{1},\hat{M}_{2})=\bigl(\delta_{1}(Y_{1}),\delta_{2}(Y_{2})\bigr)$

The rate pair $(R_{1},R_{2})$, in unit of $q$-ary symbols per subspace
transmission, of the broadcast subspace code is defined as

\begin{equation}
R_{1}=\frac{\log_{q}|\mathfrak{M}_{1}|}{n},\quad R_{2}=\frac{\log_{q}|\mathfrak{M}_{2}|}{n}.\label{eq:LOBCRate}\end{equation}

\end{defn}
As in \cite[Ch.~14.6]{cover1991elements} we can rewrite the encoding
map as \[
\gamma\colon(1,2,...,q^{nR_{1}})\times(1,2,...,q^{nR_{2}})\rightarrow\mathfrak{C}\]
 and associate with the broadcast subspace code the parameters $((q^{nR_{1}},q^{nR_{2}}),n)$. 
\begin{defn}
A rate pair $(R_{1},R_{2})$ is said to be \emph{achievable} if there
exists a sequence of $((q^{nR_{1}},q^{nR_{2}}),n)$ broadcast subspace
codes, for which the corresponding probabilities $p_{n}=p_{n}(\hat{M}_{1}\neq M_{1}\vee\hat{M}_{2}\neq M_{2})$
of decoding error satisfy $p_{n}\to0$ when $n\rightarrow\infty$.%
\footnote{Here we tacitly assume that $n$ runs through some subsequence of
the positive integers for which all numbers $q^{nR_{1}}$, $q^{nR_{2}}$
are integers.%
} The \emph{capacity region} (or \emph{rate region}) of a LOBC is defined
as the closure of the set of all achievable rate pairs. 
\end{defn}

\subsection{CMLOBCs}

If every subchannel in a LOBC is a CMLOC (necessarily with the same
$l$, cf.\ Def.~\ref{dfn:CMLOC}), we call it a \emph{constant-dimension
multiplicative LOBC (CMLOBC)}. For CMLOBCs with ambient space $\mathbb{F}_{q}^{m}$
and constant dimension $l$ the normalized rate pair $(\bar{R}_{1},\bar{R}_{2})$
can be defined in accordance with \eqref{eq:LOBCRate} as \begin{equation}
\bar{R}_{1}=\frac{\log_{q}|\mathfrak{M}_{1}|}{lmn},\quad\bar{R}_{2}=\frac{\log_{q}|\mathfrak{M}_{2}|}{lmn}.\label{eq:LOBCNorRate}\end{equation}

By the principle of time division, it is clear that the capacity region
of a CMLOBC should be at least the triangle area with three corner
points--$(0,0)$, $(C_{2},0)$ and $(0,C_{1})$ on the $(R_{1},R_{2})$
plane, where $C_{i}$ refers to the channel capacity of $\mathsf{X}\rightarrow\mathsf{Y}_{i}$,
and all points $(R_{1},R_{2})$ satisfy $R_{1}/C_{1}+R_{2}/C_{2}=1\wedge R_{1},R_{2}\geq0$
constitute the so called time sharing line.

\section{Degradation Theorem for CMLOBCs}

The following definition of degraded broadcast channels is taken from
\cite{720547}. 
\begin{defn}
A CMLOBC with transfer probabilities $p(Y_{1},Y_{2}|X)$ is said to
be \emph{(stochastically) degraded} if the conditional marginals $p(Y_{1}|X)$,
$p(Y_{2}|X)$ are related by $p(Y_{2}|X)=\sum_{Y_{1}}p(Y_{1}|X)p'(Y_{2}|Y_{1})$
for some conditional distribution $p'(Y_{2}|Y_{1})$. \label{def:CMLOBCDegraded} 
\end{defn}
From Def.~\ref{dfn:CMLOC} it is obvious that CMLOBCs with $(m,q,l)=(2,2,1)$
(the smallest nontrivial examples) are equivalent to ternary erasure
broadcast channels with erasure probabilities 
$\epsilon_{0}^{(1)}$, $\epsilon_{0}^{(2)}$ for the two subchannels.
Like a BEBC such broadcast channels are always degraded. In general,
however, 
CMLOBCs are not degraded. Theorem~\ref{thm:DCMLOBCThm} in this section
gives a necessary and sufficient condition for a CMLOBC to be degraded.
For its proof we need several lemmas. 
\begin{lem}
Let $\boldsymbol{\mathbf{\epsilon}}^{(1)}=(\epsilon_{0}^{(1)},\epsilon_{1}^{(1)},...,\epsilon_{l}^{(1)})$
and $\boldsymbol{\mathbf{\epsilon}}^{(2)}=(\epsilon_{0}^{(2)},\epsilon_{1}^{(2)},...,\epsilon_{l}^{(2)})$
be probability vectors. Then the following two statements are equivalent:

\,(i)\,\,\,\,\,$\sum_{j=0}^{i}\epsilon_{j}^{(1)}\leq\sum_{j=0}^{i}\epsilon_{j}^{(2)}$
for $0\leq i\leq l$;

(ii)\,\,\,\,\,There exists a lower triangular stochastic matrix
$\boldsymbol{\Lambda}=(\lambda_{ij})$ such that $\boldsymbol{\epsilon}^{(1)}\boldsymbol{\Lambda}=\boldsymbol{\epsilon}^{(2)}$.\label{lem:DCMLOBCLem1} 
\end{lem}
\begin{proof} See Appendix \ref{sub:ProfDCMLOBCLem1}. \end{proof}

For $0\leq l,s\leq m$ let $\mathcal{D}_{ls}$ be the incidence structure
{}``$l$-dimensional vs. $s$-dimensional subspaces of $\mathbb{F}_{q}^{m}$
with respect to set inclusion''. Relative to suitable orderings of
the input and output alphabet, the channel matrix of the CMLOC of
constant dimension $l$ with probability vector $\boldsymbol{\epsilon}=(\epsilon_{0},\epsilon_{1},\dots,\epsilon_{l})$
can be partitioned as \begin{equation}
\mathbf{S}=\bigl(\epsilon_{0}\mathbf{S}_{l0}\mid\epsilon_{1}\mathbf{S}_{l1}\mid\dots\mid\epsilon_{l}\mathbf{S}_{ll}\bigr),\label{eq:partition}\end{equation}
 where $\mathbf{S}_{ls}$ ({}``stochastic incidence matrix'' of
$\mathcal{D}_{ls}$) denotes an appropriate scalar multiple of the
incidence matrix of $\mathcal{D}_{ls}$, determined by the requirement
that $\mathbf{S}_{ls}$ be a (row) stochastic matrix.%
\footnote{The scaling factor for $\mathcal{D}_{ls}$ is $\binom{l}{s}_{q}^{-1}$.%
} 
\begin{lem}
For integers $l,s,t\in\{0,1,\dots,m\}$ with $l\geq s\geq t$ we have
$\mathbf{S}_{ls}\mathbf{S}_{st}=\mathbf{S}_{lt}$. \label{lem:DCMLOBCLem2} 
\end{lem}
\begin{proof} See Appendix \ref{sub:ProfDCMLOBCLem2}. \end{proof}
A CMLOBC with subchannels having channel matrices $\mathbf{S}^{(1)}$,
$\mathbf{S}^{(2)}$ is degraded if and only if there exists a stochastic
matrix $\mathbf{T}\in\mathbb{R}^{M\times M}$ (where $M=\sum_{s=0}^{l}\binom{m}{s}_{q}$)
such that $\mathbf{S}^{(2)}=\mathbf{S}^{(1)}\mathbf{T}$ (see \cite[Ch.~14.6]{cover1991elements}).
Partitioning $\mathbf{S}^{(1)}$, $\mathbf{S}^{(2)}$ as in \eqref{eq:partition}
and $\mathbf{T}$ accordingly, we can write this as \begin{align}
\bigl(\epsilon_{0}^{(1)}\mathbf{S}_{l0} & \mid\epsilon_{1}^{(1)}\mathbf{S}_{l1}\mid\dots\mid\epsilon_{l}^{(1)}\mathbf{S}_{ll}\bigr)\begin{pmatrix}\mathbf{T}_{00} & \mathbf{T}_{01} & \cdots & \mathbf{T}_{0l}\\
\mathbf{T}_{10} & \mathbf{T}_{11} & \cdots & \mathbf{T}_{0l}\\
\vdots & \vdots & \ddots & \vdots\\
\mathbf{T}_{l0} & \mathbf{T}_{l1} & \cdots & \mathbf{T}_{ll}\end{pmatrix}\nonumber \\
 & =\bigl(\epsilon_{0}^{(2)}\mathbf{S}_{l0}\mid\epsilon_{1}^{(2)}\mathbf{S}_{l1}|\dots\mid\epsilon_{l}^{(2)}\mathbf{S}_{ll}\bigr)\label{eq:ST}\end{align}
 With these preparations it is possible to prove 
\begin{thm}
Let $\boldsymbol{\epsilon}^{(1)}$ and $\boldsymbol{\epsilon}^{(2)}$
be probability vectors associated with the two subchannels $\mathsf{X}\to\mathsf{Y}_{1}$
and $\mathsf{X}\to\mathsf{Y}_{2}$, respectively, of a CMLOBC with
ambient space $\mathbb{F}_{q}^{m}$ and constant dimension $l<m$.
The CMLOBC is degraded (in the sense that $\mathsf{Y}_{2}$ is a degraded
version of $\mathsf{Y}_{1}$) if and only if $\epsilon^{(1)}$ and
$\epsilon^{(2)}$ satisfy \begin{equation}
\sum_{j=0}^{i}\epsilon_{j}^{(1)}\leq\sum_{j=0}^{i}\epsilon_{j}^{(2)}\quad\text{for \ensuremath{0\leq i\leq l}}.\footnotemark\label{eq:DCMLOBCCondition}\end{equation}
 \footnotetext{Note that \eqref{eq:DCMLOBCCondition} can be rewritten
as $\sum_{j=i}^{l}\epsilon_{j}^{(1)}\geq\sum_{j=i}^{l}\epsilon_{j}^{(2)}$
for $0\leq i\leq l$ and implies in particular that the probabilities
of successful transmission are related by $\epsilon_{l}^{(1)}\geq\epsilon_{l}^{(2)}$.}
\label{thm:DCMLOBCThm} 
\end{thm}
\begin{proof} See Appendix \ref{sub:ProfDCMLOBCThm}\end{proof}

The excluded case $l=m$ is indeed exceptional: In this case there
is only one input subspace, so that the channel matrices reduce to
probability vectors $\mathbf{s}^{(1)},\mathbf{s}^{(2)}$ of length
$M$, where $M=\sum_{i=0}^{m}\binom{m}{i}_{q}$ is the total number
of subspaces of $\mathbb{F}_{q}^{m}$. However any two probability
vectors $\mathbf{s}^{(1)},\mathbf{s}^{(2)}$ are related by $\mathbf{s}^{(1)}\mathbf{T}=\mathbf{s}^{(2)}$
for some stochastic matrix $\mathbf{T}$ of the appropriate size.
(The matrix $\mathbf{T}=\mathbf{j}\mathbf{s}^{(2)}$, where $\mathbf{j}$
is the all-one column vector of the same dimension as $\mathbf{s}$,
does the job.) This shows that in the case $l=m$ the broadcast channel
is degraded for all choices of $\boldsymbol{\epsilon}^{(1)}$, $\boldsymbol{\epsilon}^{(2)}$.
\begin{cor}
Under the assumptions of Th.~\ref{thm:DCMLOBCThm}, suppose that
$\boldsymbol{\epsilon}^{(1)}$ and $\boldsymbol{\epsilon}^{(2)}$
satisfy \begin{equation}
\epsilon_{i}^{(1)}\leq\epsilon_{i}^{(2)}\quad\text{for every}\quad i\in\{0,1,2,...,l-1\}.\label{eq:StrongDegradedCondition}\end{equation}
(and consequently $\epsilon_{l}^{(1)}\geq\epsilon_{l}^{(2)}$) Then
the CMLOBC is degraded (in the sense that $\mathsf{Y}_{2}$ is a degraded
version of $\mathsf{Y}_{1}$) . \label{cor:DCMLOBCCoro1} 
\end{cor}

\section{The Capacity Region of Degraded CMLOBCs over the Projective Plane
$\operatorname{PG}(2,2)$}

\subsection{Degraded CMLOBCs over the Projective Plane $\operatorname{PG}(2,2)$}

Let $q=2$, $m=3$, $l=2$ and $p(Y_{i}|X)$, $i=1,2$, be defined
through the channel matrices \begin{equation}
\mathbf{S}^{(i)}=\bigl(\epsilon_{0}^{(i)}\mathbf{J}_{7\times1}\mid\epsilon_{1}^{(i)}\mathbf{S}_{21}\mid\epsilon_{2}^{(i)}\mathbf{I}_{7\times7}\bigr),\label{eq:WCMLOBExa}\end{equation}
 where $\mathbf{J}_{7\times1}$, $\mathbf{I}_{7\times7}$ denote the
all-one, respectively, the identity matrix of the indicated sizes
and $\mathbf{S}_{21}$ is a stochastic incidence matrix of $2$-dimensional
vs. $1$-dimensional subspaces of $\mathbb{F}_{2}^{3}$ (in other
words, an incidence matrix of the smallest projective plane $\operatorname{PG}(2,2)$).
For example, we can take \begin{equation}
\mathbf{S}_{21}=\frac{1}{3}\left(\begin{array}{ccccccc}
1 & 1 & 0 & 1 & 0 & 0 & 0\\
1 & 0 & 1 & 0 & 1 & 0 & 0\\
0 & 1 & 1 & 0 & 0 & 1 & 0\\
0 & 0 & 1 & 1 & 0 & 0 & 1\\
0 & 1 & 0 & 0 & 1 & 0 & 1\\
1 & 0 & 0 & 0 & 0 & 1 & 1\\
0 & 0 & 0 & 1 & 1 & 1 & 0\end{array}\right).\label{eq:S21}\end{equation}

By Th.~\ref{thm:DCMLOBCThm} the CMLOBC is degraded if and only if
$\epsilon_{0}^{(1)}\leq\epsilon_{0}^{(2)}\wedge\epsilon_{0}^{(1)}+\epsilon_{1}^{(1)}\leq\epsilon_{0}^{(2)}+\epsilon_{1}^{(2)}$
or, equivalently, $\epsilon_{0}^{(1)}\leq\epsilon_{0}^{(2)}\wedge\epsilon_{2}^{(1)}\geq\epsilon_{2}^{(2)}$.

Taking into account symmetry properties and keeping in mind the example
of binary symmetric broadcast channels discussed in \cite[Ch.~14.6]{cover1991elements},
one might conjecture that the boundary of the rate region is obtained
by taking the joint distribution $p(U,X)$ which arises from a $7$-ary
symmetric channel $\mathsf{U}\to\mathsf{X}$ and the uniform input
distribution on $\mathfrak{U}$. This one-parameter family of distributions
can be written in matrix form as \begin{align}
\bigl(p(U_{i},X_{j})\bigr) & =\tfrac{1}{7}\left(\tfrac{\sigma}{6}\mathbf{J}_{7\times7}+\left(1-\tfrac{7\sigma}{6}\right)\mathbf{I}_{7\times7}\right),\quad0\leq\sigma\leq\tfrac{6}{7}.\label{eq:pxu1}\end{align}

\begin{lem}
For the degraded CMLOBCs described by \eqref{eq:WCMLOBExa}, let $p_{\mathsf{U},\mathsf{X}}(U,X)$
be chosen as in \eqref{eq:pxu1}, and with $R_{1}(\sigma)=I(\mathsf{X},\mathsf{Y}_{1}|\mathsf{U})$,
$R_{2}(\sigma)=I(\mathsf{U},\mathsf{Y}_{2})$ let $\Gamma=\{\bigl(R_{1}(\sigma),R_{2}(\sigma)\bigr)\mid\sigma\in[0,6/7]\}$.
Then the curve $\Gamma$, considered as a function $R_{2}=f(R_{1})$
is defined on $[0,C_{1}]$, strictly decreasing, and satisfies $f(0)=C_{2}$,
$f(C_{1})=0$. Further we have:

(i) $f$ is strictly concave ($\cap$) when $\epsilon_{1}^{(1)}\epsilon_{2}^{(2)}>\epsilon_{1}^{(2)}\epsilon_{2}^{(1)}$;

(ii) $f$ is strictly convex ($\cup$) when $\epsilon_{1}^{(1)}\epsilon_{2}^{(2)}<\epsilon_{1}^{(2)}\epsilon_{2}^{(1)}$;

(iii) $f$ is linear (i.e.\ $\Gamma$ coincides with the time-sharing
line) when $\epsilon_{1}^{(1)}\epsilon_{2}^{(2)}=\epsilon_{1}^{(2)}\epsilon_{2}^{(1)}$.
\label{lem:WCMLOBCLem1} 
\end{lem}
\begin{proof} See Appendix \ref{sub:ProfWCMLOBCLem1}.\end{proof} 
\begin{rem}
If Case~(i) holds for a degraded CMLOBC, then there exist superposition
coding schemes which are superior to time sharing with respect to
channel throughput. On the other hand, the family of joint distributions~\eqref{eq:pxu1}
does not necessarily determine the boundary of the capacity region.
In particular we cannot conclude that in Case (ii) or (iii) of Lemma~\ref{lem:WCMLOBCLem1}
the boundary is the time-sharing line. \label{rmk:WCMLOBCLem1} 
\end{rem}

\subsection{Numerical Analysis}

In each figure about the capacity region of some CMLOBC, we use a
``\textit{filter}'' to delete points located below the time sharing
line on the rate region plane and we always display two subfigures
``\textit{before filter}'' and ``\textit{after filter}'' at the same
time. All the figures have enough pixel information to allow enlarging
details. Relevant M-files can be found at \cite{Panga}. Our analysis
was done using MATLAB on a Linux system. 
\begin{example}
Let $\boldsymbol{\epsilon}^{(1)}=(0.05,0.24,0.71)$, $\boldsymbol{\epsilon}^{(2)}=(0.30,0.15,0.55)$.
Then the condition of Case (i) is satisfied. Numerical results obtained
by using the Arimoto-Blahut type algorithm from \cite{5513525} are
shown in Fig.\,\ref{fig:CREx11}. 
\end{example}
\begin{figure}[h]
\subfloat[before filter]{\includegraphics[scale=0.28]{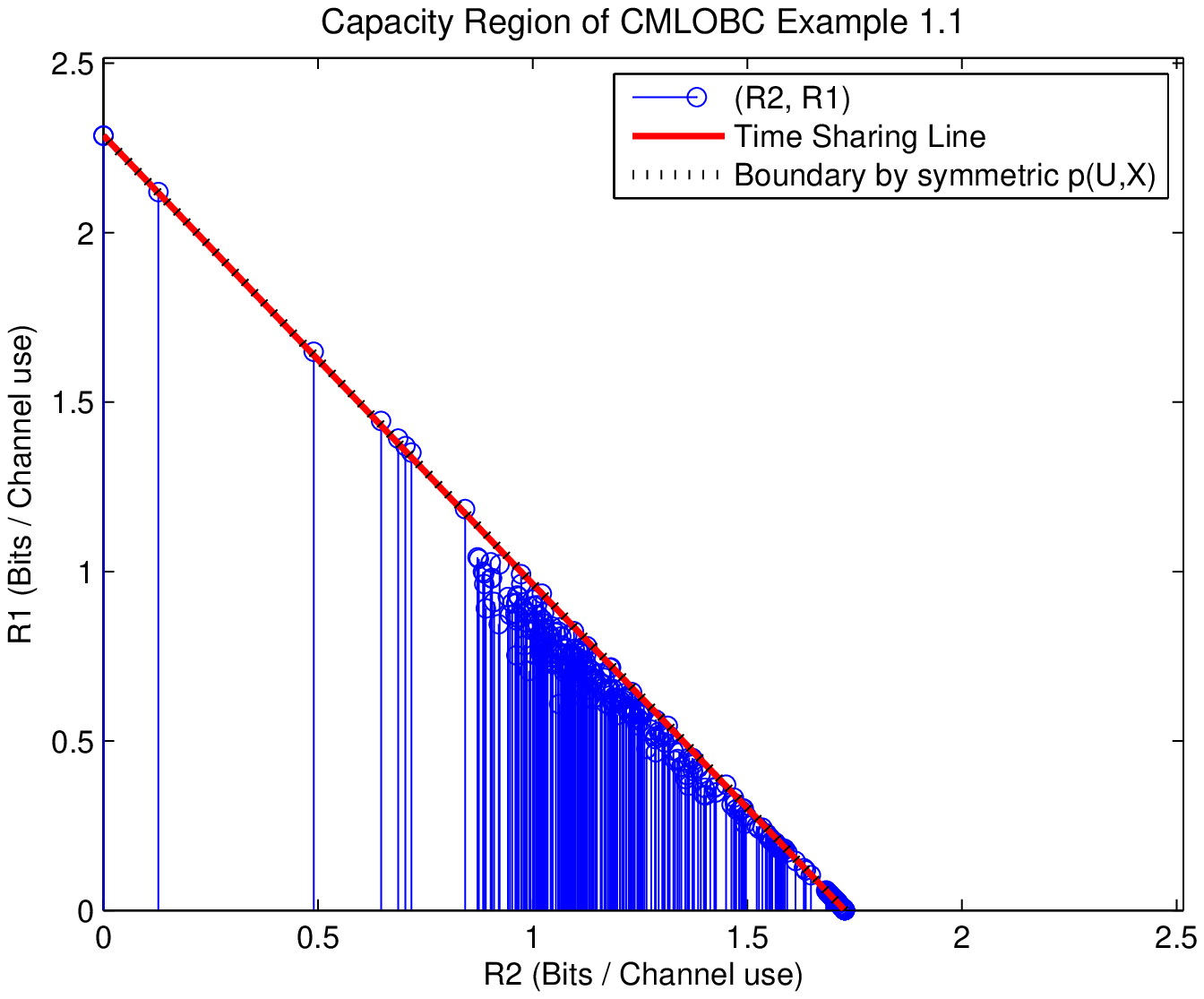}}\subfloat[after filter]{\includegraphics[scale=0.28]{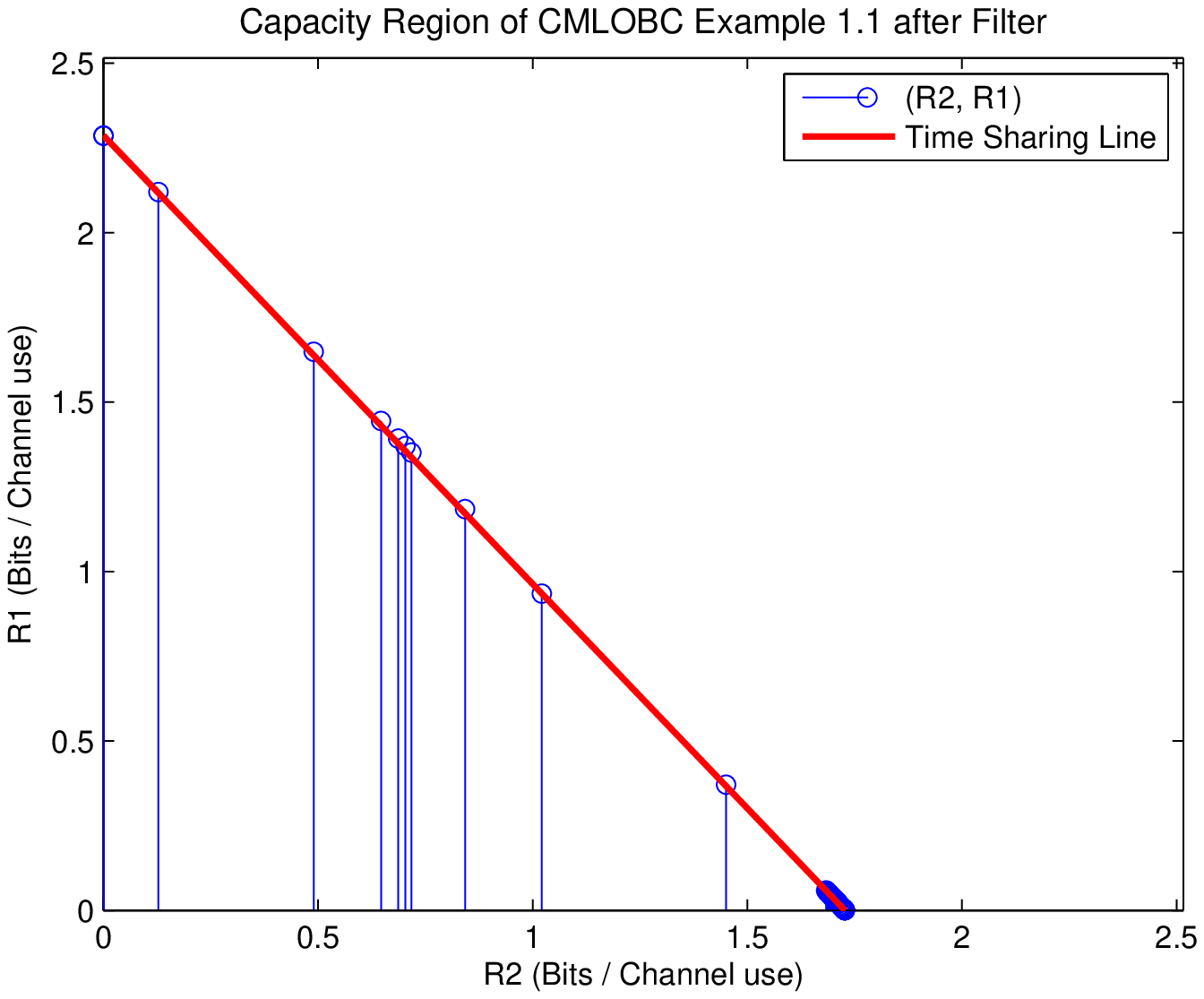}}\caption{Capacity region of Example 1, $\boldsymbol{\epsilon}^{(1)}=(0.05,0.24,0.71)$,
$\boldsymbol{\epsilon}^{(2)}=(0.30,0.15,0.55)$. \label{fig:CREx11}}

\end{figure}

\begin{example}
Let $\boldsymbol{\epsilon}^{(1)}=(0.05,0.20,0.75)$, $\boldsymbol{\epsilon}^{(2)}=(0.30,0.15,0.55)$,
the condition of case (ii) is satisfied. Numerical results are shown
in Fig.\,\ref{fig:CREx12} indicating that time sharing might be
sufficient to exhaust the capacity region.
\end{example}
\begin{figure}[h]
\subfloat[before filter]{\includegraphics[scale=0.28]{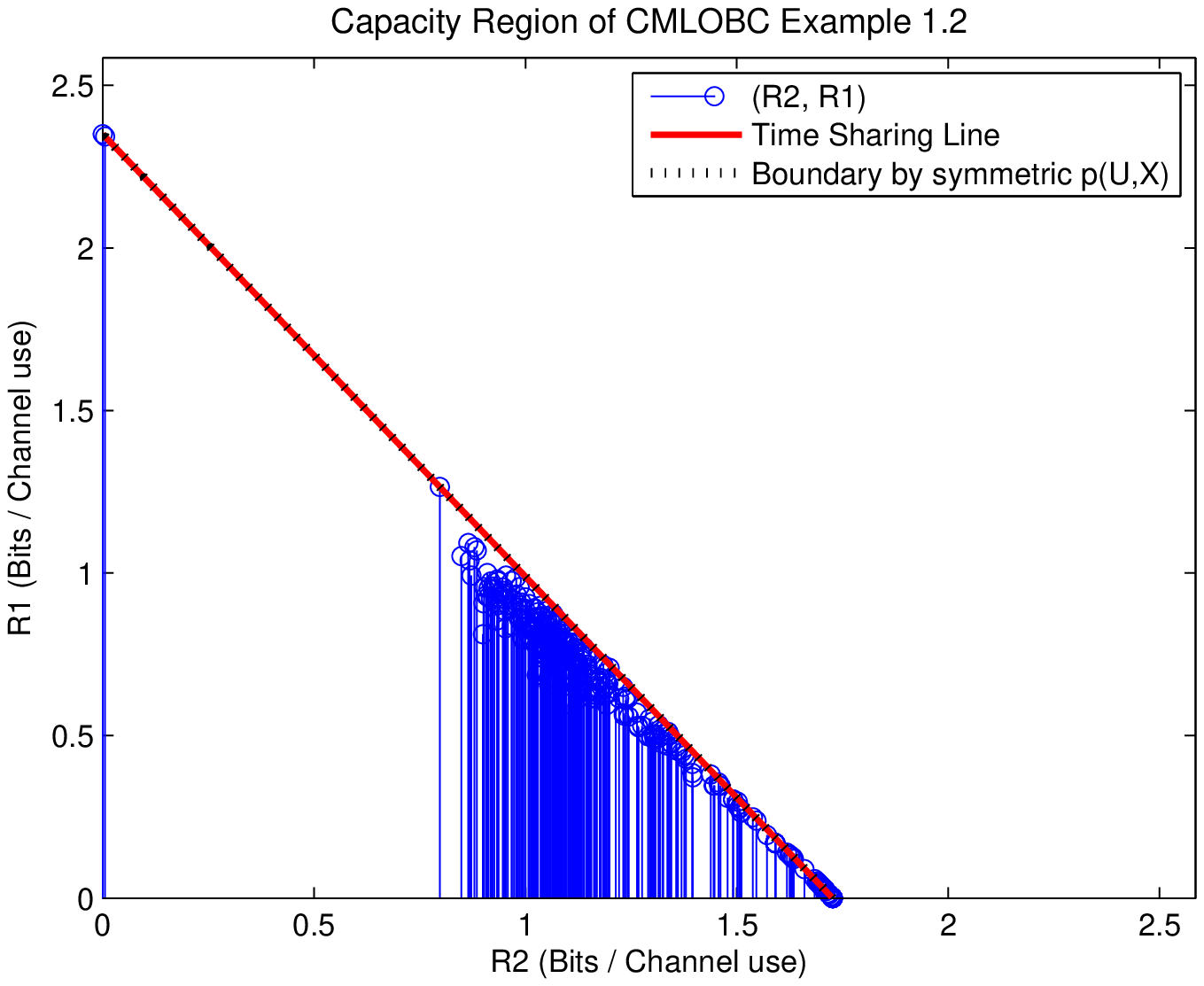}}\subfloat[after filter]{\includegraphics[scale=0.28]{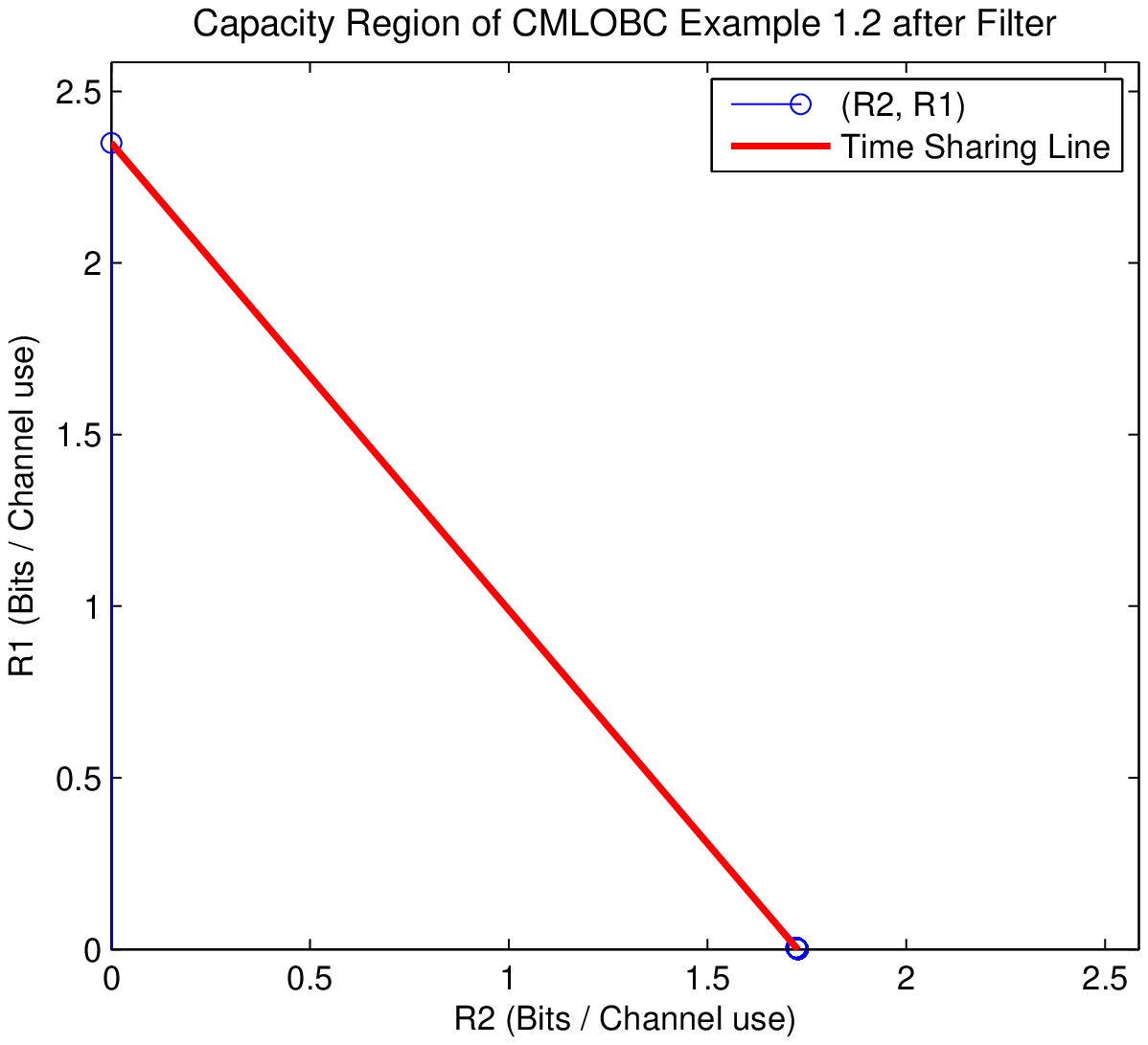}}\caption{Capacity region of Example 2, $\boldsymbol{\epsilon}^{(1)}=(0.05,0.20,0.75)$,
$\boldsymbol{\epsilon}^{(2)}=(0.30,0.15,0.55)$. \label{fig:CREx12}}

\end{figure}

\begin{example}
Let $\boldsymbol{\epsilon}^{(1)}=(\rho_{1}^{2},\rho_{1},1-\rho_{1}-\rho_{1}^{2})$,
$\boldsymbol{\epsilon}^{(2)}=(\rho_{2}^{2},\rho_{2},1-\rho_{2}-\rho_{2}^{2})$,
where $0\leq\rho_{1}\leq\rho_{2}\leq(-1+\sqrt{5})/2$. This corresponds
to Case (ii). Numerical results are shown in Fig.\,\ref{fig:CREx2},
for the particular case $\rho_{1}=0.1$, $\rho_{2}=0.3$ indicating
that time sharing might be sufficient to exhaust the capacity region. 
\end{example}
\begin{figure}[h]
 \subfloat[before filter]{\includegraphics[scale=0.28]{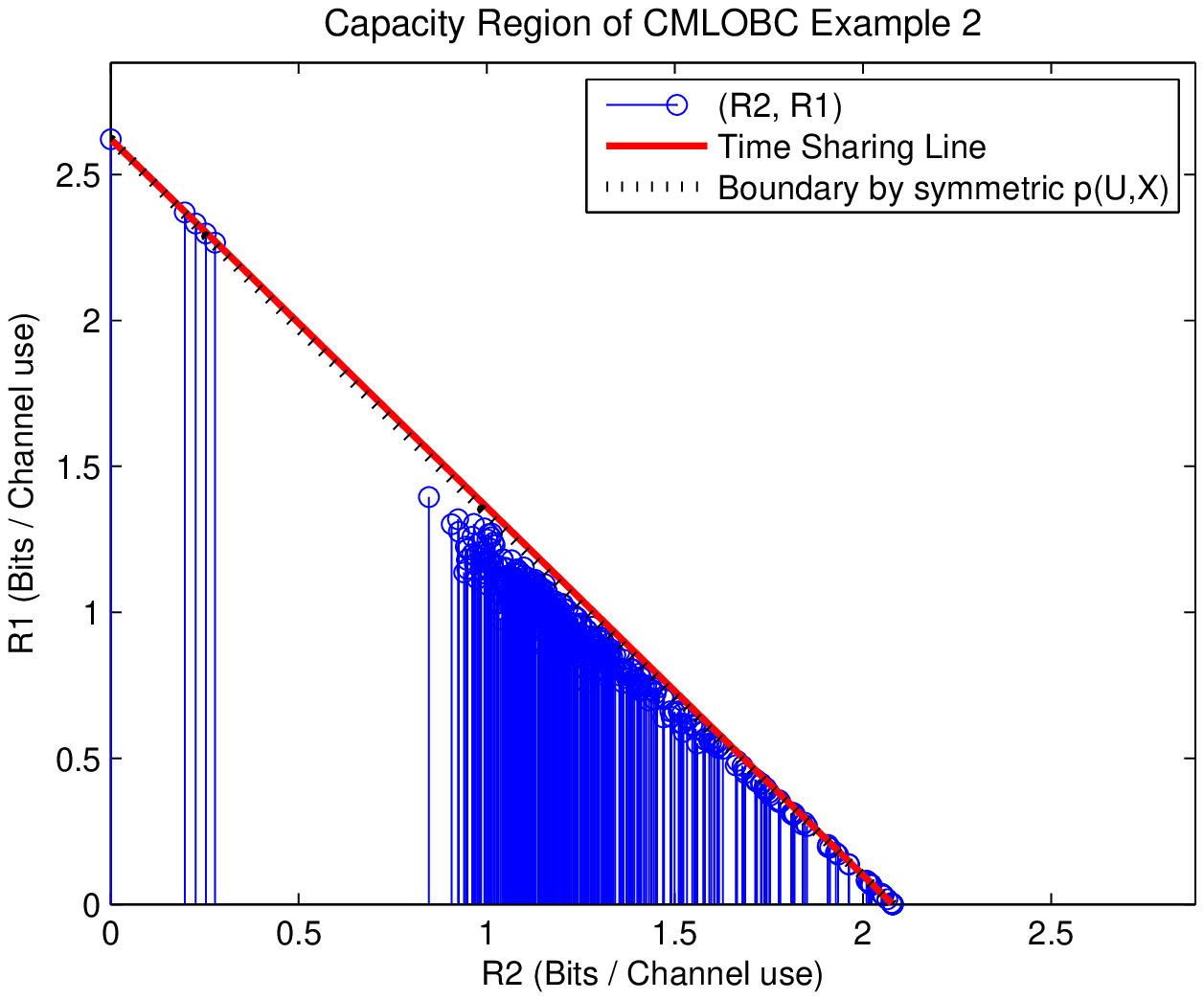}}\subfloat[after filter]{\includegraphics[scale=0.28]{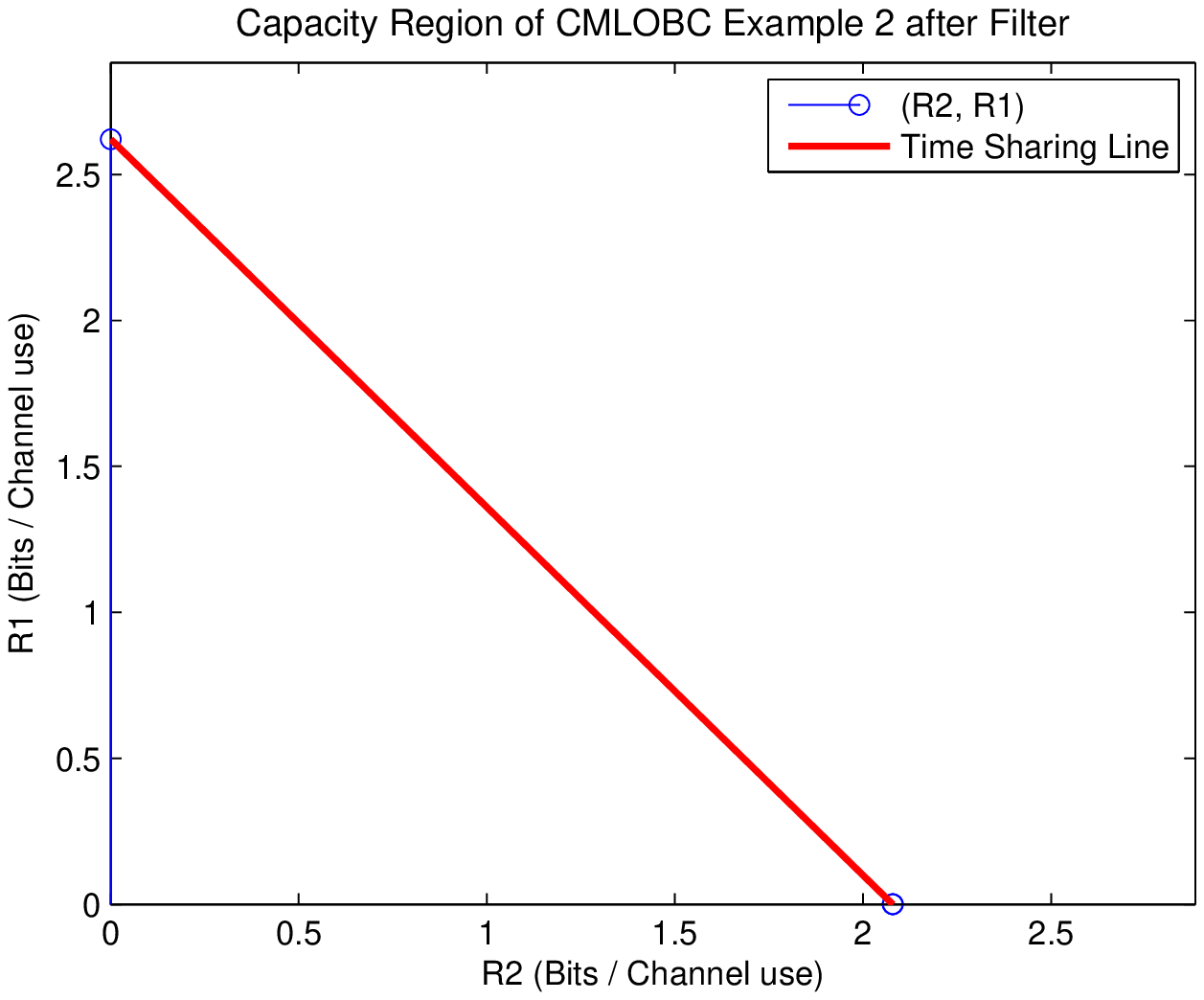}}\caption{Capacity Region of Example 3, $\boldsymbol{\epsilon}^{(1)}=(0.01,0.1,0.89)$,
$\boldsymbol{\epsilon}^{(2)}=(0.09,0.3,0.61)$\label{fig:CREx2}}

\end{figure}

\begin{example}
let $q=2$, $m=3$, $l=2$, and define $\boldsymbol{\epsilon}^{(1)}=(0,\rho_{1},1-\rho_{1})$,
$\boldsymbol{\epsilon}^{(2)}=(0,\rho_{2},1-\rho_{2})$. with $0\leq\rho_{1}\leq\rho_{2}\leq1$,
the condition of case (ii) is satisfied. Numerical results are shown
in Fig.\,\ref{fig:CREx3} for the particular case $\rho_{1}=0.1$,
$\rho_{2}=0.3$ indicating that time sharing is suffice to exhaust
the capacity region.
\end{example}
\begin{figure}[h]
\subfloat[before filter]{\includegraphics[scale=0.28]{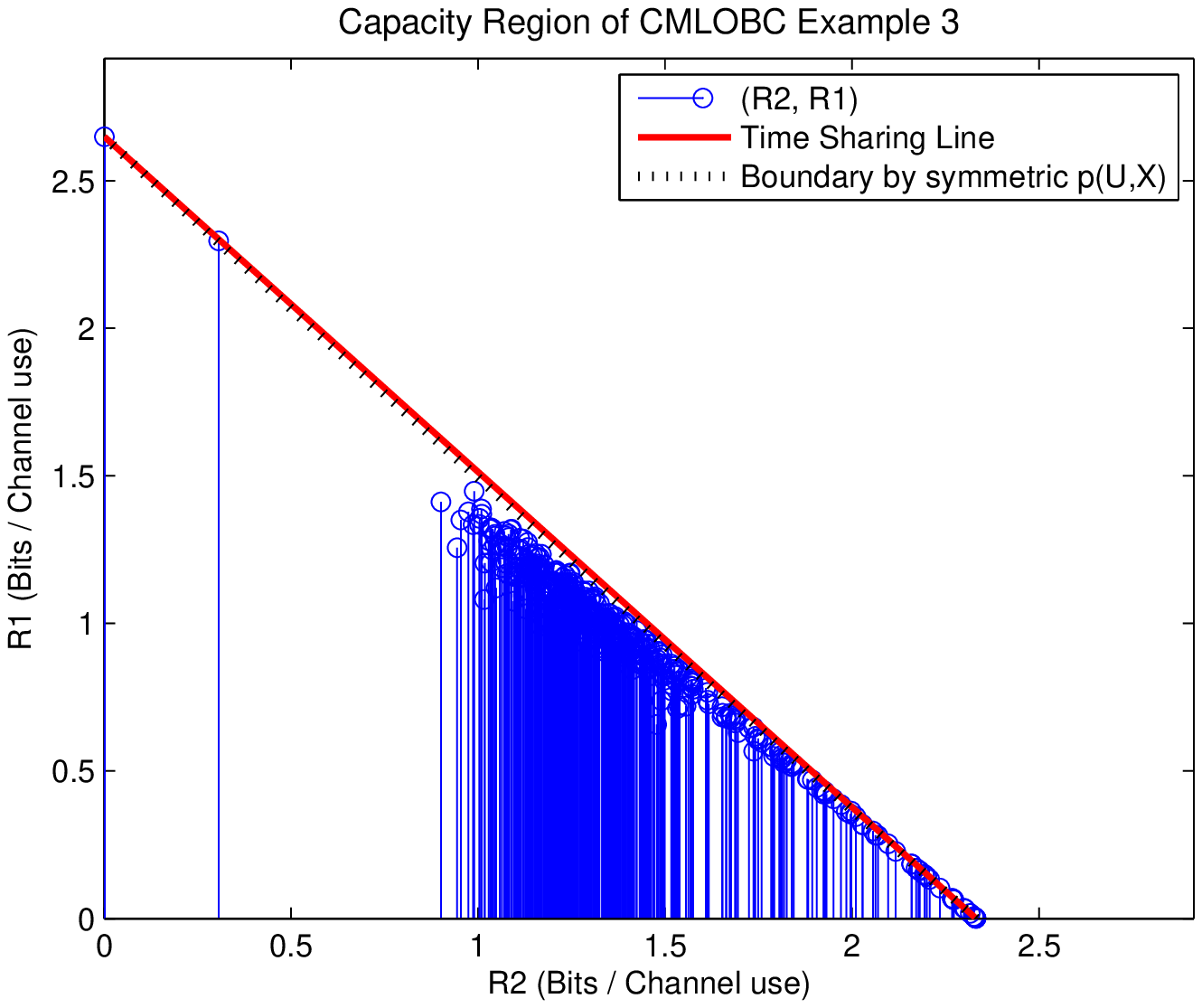}}\subfloat[after filter]{\includegraphics[scale=0.28]{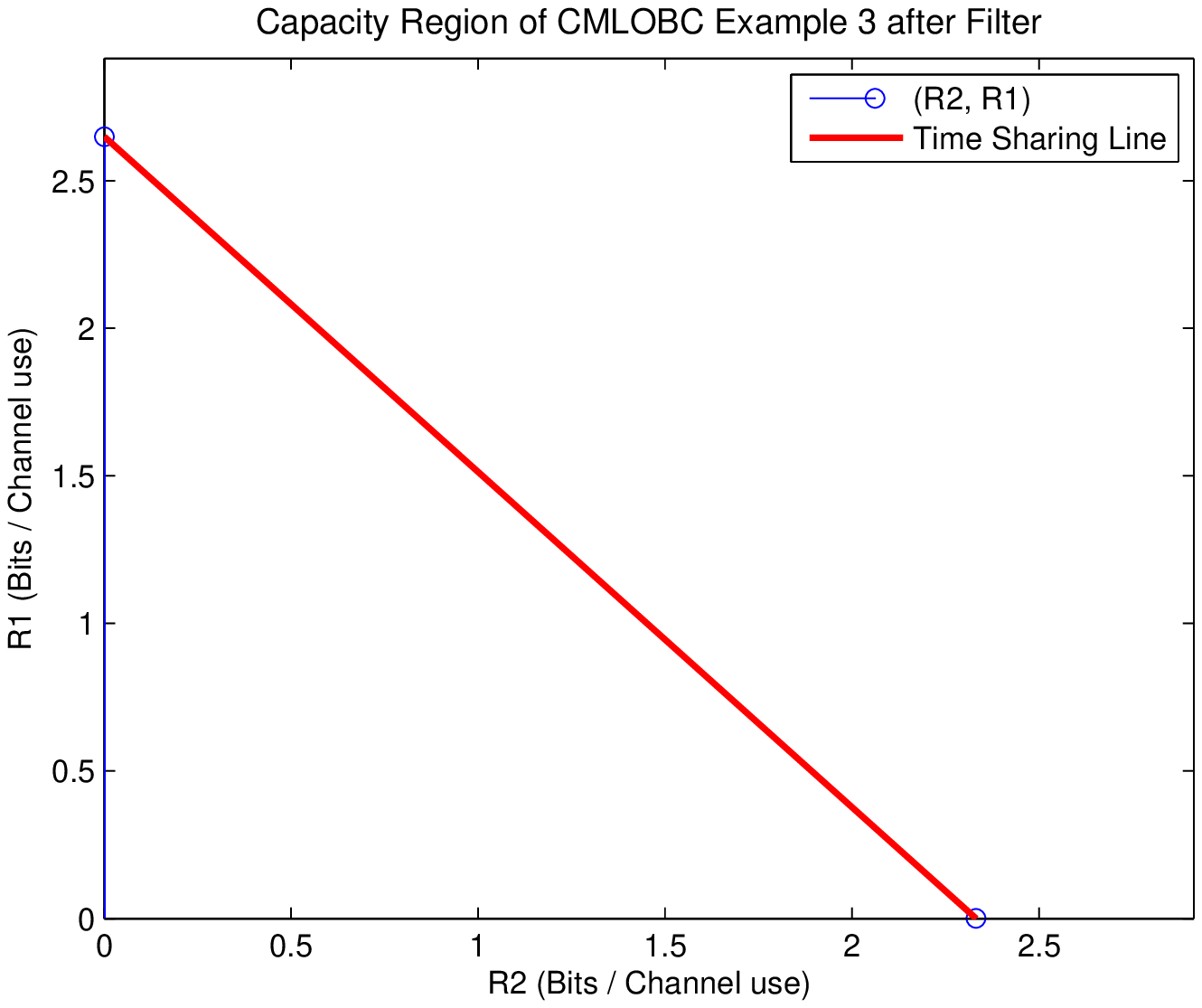}}\caption{Capacity Region of Example 4, $\boldsymbol{\epsilon}^{(1)}=(0,0.1,0.9)$,
$\boldsymbol{\epsilon}^{(2)}=(0,0.3,0.7)$\label{fig:CREx3}}

\end{figure}

\subsection{A Conjecture on the Convexity of Capacity Region }

Overall the analysis supports the conclusion that superposition coding
on CMLOBCs has no benefit over simple time-sharing unless we are in
Case (i). However, proving the conjecture in full generality seems
to be difficult. 
\begin{conjecture*}
For the degraded CMLOBCs described by \eqref{eq:WCMLOBExa}, the capacity
region is strictly concave ($\cap$) if and only if $\epsilon_{1}^{(1)}\epsilon_{2}^{(2)}>\epsilon_{1}^{(2)}\epsilon_{2}^{(1)}$.
\end{conjecture*}

\subsection{A special example--The 7-ary erasure broadcast channel}
\begin{example}
Let $\boldsymbol{\epsilon}^{(1)}=(\rho_{1},0,1-\rho_{1})$, $\boldsymbol{\epsilon}^{(2)}=(\rho_{2},0,1-\rho_{2})$,
where $0\leq\rho_{1}\leq\rho_{2}\leq1$. Then the condition of Case
(iii) is satisfied. Since (apart from unused output subspaces) there
is now only one erasure symbol (the output subspace $\{\mathbf{0}\}$),
the subchannels of the CMLOBC become $7$-ary erasure channels. 
\end{example}
The capacity region of this broadcast channel, more generally of any
CMLOBC with $\boldsymbol{\epsilon}^{(i)}=(\rho_{i},0,...,0,1-\rho_{i})$
for $i=1,2$, where $0\leq\rho_{1}\leq\rho_{2}\leq1$, is determined
by the next theorem. For the proof we need the following lemma. 
\begin{lem}
Let $\mathsf{U}$, $\mathsf{X}$ and $\mathsf{Y}$ be random variables
with alphabets $\mathfrak{U}$, $\mathfrak{X}$ and $\mathfrak{Y}$,
respectively, forming a Markov chain $\mathsf{U\rightarrow\mathsf{X}\rightarrow\mathsf{Y}}$.
Suppose that $\mathsf{X}\rightarrow\mathsf{Y}$ is described by

\begin{equation}
\bigl(p(Y_{j}|X_{i})\bigr)=\left(\begin{array}{ccc}
\rho\mathbf{J}_{|\mathfrak{X}|\times1} & | & (1-\rho)\mathbf{I}_{|\mathfrak{X}|\times|\mathfrak{X}|}\end{array}\right).\label{eq:q_ary_Channel_transfer}\end{equation}
 Then we have the relationships \begin{equation}
I(\mathsf{U};\mathsf{Y})=(1-\rho)I(\mathsf{U};\mathsf{X}),\label{eq:q_ary_Relation1}\end{equation}
 \begin{equation}
I(\mathsf{X};\mathsf{Y}|\mathsf{U})=(1-\rho)I(\mathsf{X};\mathsf{X}|U).\label{eq:q_ary_Relation2}\end{equation}
 \label{lem:q_ary_Lem} 
\end{lem}
This follows from linearity of mutual information with respect to
the decomposition \eqref{eq:q_ary_Channel_transfer} and $I(\mathsf{X};\mathsf{Y}|\mathsf{U})=I(\mathsf{X};\mathsf{Y})-I(\mathsf{U};\mathsf{Y})$. 
\begin{thm}
Suppose that the two subchannels of a CMLOBC are described by \begin{equation}
\mathbf{S}^{(i)}=\bigl(\rho_{i}\mathbf{J}_{|\mathfrak{X}|\times1}\mid(1-\rho_{i})\mathbf{I}_{|\mathfrak{X}|\times|\mathfrak{X}|}\bigr),\label{eq:q_ary_CMLOBC}\end{equation}
 where $\mathfrak{X}=\mathcal{P}(\mathbb{F}_{q}^{m},l)$ and $0\leq\rho_{1}\leq\rho_{2}\leq1$.
Then its capacity region is the set of all pairs of $(R_{1},R_{2})$
satisfying $R_{1},R_{2}\geq0$ and \begin{equation}
\frac{R_{1}}{(1-\rho_{1})\log|\mathfrak{X}|}+\frac{R_{2}}{(1-\rho_{2})\log|\mathfrak{X}|}\leq1.\label{eq:q_ary_CapacityRegion}\end{equation}
 \label{thm:q_ary_Capacity_Thm}
\end{thm}
\begin{proof} See Appendix \ref{sub:Profq_ary_Capacity_Thm}.\end{proof}

\section{Conclusion}

In this paper, we have set up the framework of linear operator broadcast
channels. We characterized degraded CMLOBCs by a set of inequalities
for their associated probability vectors. Necessary and sufficient
conditions for a CMLOBC being degraded were obtained. The work on
CMLOBCs over $\operatorname{PG}(2,2)$ shows that time sharing schemes
do not always exhaust the capacity region.

We conclude with some open problems arising from our work. 
\begin{itemize}
\item In the case of more general CMLOBCs (i.e. less noisy, more capable),
whose rate region is not exhausted by superposition coding, investigate
whether other coding technologies (dirty paper coding, etc.) are suitable
for approaching the boundary. 
\item How does the rate region of additive LOBCs or even more general LOBCs
look like? The example of the binary symmetric broadcast channel suggests
that in the generic case the nontrivial boundary curve $R_{2}=f(R_{1})$
is given by a strictly concave ($\cap$) function. 
\item Construct good (multishot) superposition subspace codes for degraded
LOBCs in the case, where rate splitting is needed to approach the
boundary of the rate region. 
\end{itemize}

\section{Acknowledgments}

We wish to thank Prof.\ Ning Cai, Xidian University, Xi'an, China
for helpful discussions and valuable suggestions for the proof of
the outer bound of CMLOBCs. We are indebted to Kensuke Yasui, Hitachi
Ltd., Japan for mailing us a Java script implementing the Arimoto-Blahut
type algorithm.

\section{Appendix}

\subsection{Proof of Lemma\,\ref{lem:DCMLOBCLem1}\label{sub:ProfDCMLOBCLem1}}

\begin{proof} Suppose first that (ii) holds. Postmultiplying the
equation $\boldsymbol{\epsilon}^{(1)}\boldsymbol{\Lambda}=\boldsymbol{\epsilon}^{(2)}$
by the matrix \begin{equation}
\mathbf{L}=\left(\begin{array}{cccc}
1 & 0 & \cdots & 0\\
1 & 1 & \cdots & 0\\
\vdots & \vdots & \ddots & 0\\
1 & 1 & \cdots & 1\end{array}\right)\end{equation}
 we obtain $\boldsymbol{\epsilon}^{(1)}\boldsymbol{\Lambda}\mathbf{L}=\boldsymbol{\epsilon}^{(2)}\mathbf{L}$.
The matrix $\boldsymbol{\Delta}=\boldsymbol{\Lambda}\mathbf{L}=(\delta_{ij})$
is lower triangular with entries $\delta_{ij}\leq1$. (This follows
from $\delta_{ij}=\sum_{k=0}^{l}\lambda_{ik}l_{kj}=\sum_{k=j}^{l}\lambda_{ik}\leq\sum_{k=0}^{l}\lambda_{ik}=1$.)
Hence we have \begin{align*}
\sum_{i=j}^{l}\epsilon_{i}^{(2)} & =(\boldsymbol{\epsilon}^{(2)}\mathbf{L})_{j}=(\boldsymbol{\epsilon}^{(1)}\boldsymbol{\Delta})_{j}\\
 & =\sum_{i=j}^{l}\epsilon_{i}^{(1)}\delta_{ij}\leq\sum_{i=j}^{l}\epsilon_{i}^{(1)}\,\,\,\,\,(0\leq j\leq l).\end{align*}
 Then\[
\sum_{i=0}^{j}\epsilon_{i}^{(1)}=1-\sum_{i=j}^{l}\epsilon_{i}^{(1)}\leq1-\sum_{i=j}^{l}\epsilon_{i}^{(2)}=\sum_{i=0}^{j}\epsilon_{i}^{(2)}\]
 which implies (i).

Now suppose that (i) holds. First we consider the special case where
$\boldsymbol{\epsilon}^{(1)}$ and $\boldsymbol{\epsilon}^{(2)}$
are related in the following way: There exist $0\leq i<j\leq l$ and
a real number $0\leq\lambda\leq1$ such that $\epsilon_{i}^{(2)}=\epsilon_{i}^{(1)}+\lambda\epsilon_{j}^{(1)}$,
$\epsilon_{j}^{(2)}=(1-\lambda)\epsilon_{j}^{(1)}$ and $\epsilon_{k}^{(1)}=\epsilon_{k}^{(2)}$
for $k\in\{0,1,\dots,l\}\setminus\{i,j\}$. In this case we have $\boldsymbol{\epsilon}^{(1)}\boldsymbol{\Lambda}=\boldsymbol{\epsilon}^{(2)}$,
where $\boldsymbol{\Lambda}$ differs from the identity matrix only
in the submatrix corresponding to rows and columns No. $i$, $i+1$,
\ldots{}, $j$. The corresponding submatrix of $\boldsymbol{\Lambda}$
is \begin{equation}
\begin{pmatrix}1\\
 & \ddots\\
 &  & 1\\
\lambda &  &  & 1-\lambda\end{pmatrix},\end{equation}
 so that $\boldsymbol{\Lambda}$ is clearly lower triangular and stochastic.
In general, as is easily proved by induction, a new $\boldsymbol{\epsilon}^{(2)}$
can be updated from $\boldsymbol{\epsilon}^{(1)}$ and last $\boldsymbol{\epsilon}^{(2)}$
by a sequence of transformations of the above form (i.e., add $\lambda$
times the $j$-th component to the $i$-th component and subtract
it from the $j$-th component for some $0\leq i<j\leq l$ and $0\leq\lambda\leq1$).
Since the set of lower triangular stochastic matrices is closed under
matrix multiplication, the result follows. \end{proof}

\subsection{Proof of Lemma\,\ref{lem:DCMLOBCLem2}\label{sub:ProfDCMLOBCLem2}}

\begin{proof} Working with the ordinary incidence matrices $\mathbf{D}_{ls}$,
$\mathbf{D}_{st}$, $\mathbf{D}_{lt}$, the $(i,j)$-entry of $\mathbf{D}_{ls}\mathbf{D}_{st}$
is equal to the number of subspaces $V\in\mathcal{P}(\mathbb{F}_{q}^{m},s)$
satisfying $U_{i}\supseteq V\supseteq W_{j}$, where $U_{i}\in\mathcal{P}(\mathbb{F}_{q}^{m},l)$
and $W_{j}\in\mathcal{P}(\mathbb{F}_{q}^{m},t)$ denote the $i$-th
resp.\ $j$-th subspace in the given ordering on $\mathcal{P}(\mathbb{F}_{q}^{m},l)$
resp. $\mathcal{P}(\mathbb{F}_{q}^{m},t)$. Thus \begin{equation}
(\mathbf{D}_{ls}\mathbf{D}_{st})_{ij}=\begin{cases}
\binom{l-t}{s-t}_{q} & \text{if \ensuremath{U_{i}\supseteq W_{j}}},\\
0 & \text{if \ensuremath{U_{i}\nsupseteq W_{j}}}.\end{cases}\end{equation}
 This shows that $\mathbf{D}_{ls}\mathbf{D}_{st}=\binom{l-t}{s-t}_{q}\mathbf{D}_{lt}$
is a scalar multiple of $\mathbf{D}_{lt}$. Obviously we then also
have $\mathbf{S}_{ls}\mathbf{S}_{st}=\lambda\mathbf{S}_{lt}$ for
some scalar $\lambda$. Since $\mathbf{S}_{ls}\mathbf{S}_{st}$ as
well as $\mathbf{S}_{lt}$ are stochastic, we must have $\lambda=1$,
proving the lemma. \end{proof}

\subsection{Proof of Theorem\,\ref{thm:DCMLOBCThm}\label{sub:ProfDCMLOBCThm}}

\begin{proof} Suppose first that Condition~\eqref{eq:DCMLOBCCondition}
is satisfied. In \eqref{eq:ST} we choose $\mathbf{T}_{ij}=\lambda_{ij}\mathbf{S}_{ij}$
with $\lambda_{ij}\in\mathbb{R}$ (where it is understood that $\mathbf{S}_{ij}=\mathbf{0}$
whenever $i<j$). Using Lemma~\ref{lem:DCMLOBCLem2} we obtain \begin{align}
\mathbf{S}^{(1)}\mathbf{T} & =\bigl(\epsilon_{0}^{(1)}\mathbf{S}_{l0}\mid\epsilon_{1}^{(1)}\mathbf{S}_{l1}\mid\dots\mid\epsilon_{l}^{(1)}\mathbf{S}_{ll}\bigr)\times\nonumber \\
 & \begin{pmatrix}\lambda_{00}\mathbf{S}_{00} & 0 & 0 & 0\\
\lambda_{10}\mathbf{S}_{10} & \lambda_{11}\mathbf{S}_{11} & \cdots & 0\\
\vdots & \vdots & \ddots & \vdots\\
\lambda_{l0}\mathbf{S}_{l0} & \lambda_{l1}\mathbf{S}_{l1} & \cdots & \lambda_{ll}\mathbf{S}_{ll}\end{pmatrix}\\
 & =\bigl((\epsilon_{0}^{(1)}\lambda_{00}+\epsilon_{1}^{(1)}\lambda_{10}+\dots+\epsilon_{l}^{(1)}\lambda_{l0})\mathbf{S}_{l0},\nonumber \\
 & \,\,\,\,\,\,\,\,(\epsilon_{1}^{(1)}\lambda_{11}+\epsilon_{1}^{(1)}\lambda_{21}\dots+\epsilon_{l}^{(1)}\lambda_{l1})\mathbf{S}_{l1},\dots,\epsilon_{l}^{(1)}\lambda_{ll}\mathbf{S}_{ll}\bigr)\end{align}
 By Lemma~\ref{lem:DCMLOBCLem1} we can further choose $\boldsymbol{\Lambda}=(\lambda_{ij})$
as a lower triangular stochastic matrix satisfying $\boldsymbol{\epsilon}^{(1)}\boldsymbol{\Lambda}=\boldsymbol{\epsilon}^{(2)}$.
Then the resulting matrix $\mathbf{T}=(\lambda_{ij}\mathbf{S}_{ij})$
is stochastic and satisfies \eqref{eq:ST}. Hence in this case the
broadcast channel is degraded.

Conversely suppose the broadcast channel is degraded, so that \eqref{eq:ST}
holds for some stochastic (block) matrix $\mathbf{T}=(\mathbf{T}_{ij})$.
First we will show that we can assume (without loss of generality)
that $\mathbf{T}_{ij}=0$ for $i<j$. \eqref{eq:ST} says \[
\sum_{i=0}^{l}\epsilon_{i}^{(1)}(\mathbf{S}_{li}\mathbf{T}_{ij})=\epsilon_{j}^{(2)}\mathbf{S}_{lj}\quad\text{for \ensuremath{0\leq j\leq l}}\]
 If $\epsilon_{i}^{(1)}=0$ then we can replace each block $\mathbf{T}_{ij}$,
$0\leq j\leq l$, by the corresponding all-zero matrix. Hence the
assertion is true in this case. On the other hand, if $\epsilon_{i}^{(1)}>0$
then every positive entry in $\mathbf{S}_{li}\mathbf{T}_{ij}$ forces
a positive entry of $\mathbf{S}_{lj}$ in the same position. Now suppose
$\mathbf{T}_{ij}$ has a nonzero (i.e.\ positive) entry in a position
indexed by some subspaces $V\in\mathcal{P}(\mathbb{F}_{q}^{m},i)$,
$W\in\mathcal{P}(\mathbb{F}_{q}^{m},j)$. Then $\mathbf{S}_{li}\mathbf{T}_{ij}$
has a positive entry in each position indexed by the same subspace
$W$ (as a column index) and any subspace $U\in\mathcal{P}(\mathbb{F}_{q}^{m},l)$
which contains $V$ (as a row index).

If $i<j$ then we can find a subspace $U\in\mathcal{P}(\mathbb{F}_{q}^{m},l)$
which contains $V$ but not $W$. This can be seen as follows: The
space $\overline{W}=(W+V)/V$ is a nonzero subspace of $\mathbb{F}_{q}^{m}/V$.
Hence there exists a subspace $\overline{U}$ of $\mathbb{F}_{q}^{m}/V$
of dimension $l-i<m-i$ which does not contain $\overline{W}$. Then
the preimage $U$ of $\overline{U}$ in $\mathbb{F}_{q}^{m}$ has
the required property. (The assumption $l<m$ is essential here!)

Since $U$ contains $V$ but not $W$, the matrix $\mathbf{S}_{li}\mathbf{T}_{ij}$
has an entry $>0$ in the position corresponding to $(U,W)$ and $\mathbf{S}_{lj}$
has a zero in this position. This contradiction shows that $\epsilon_{i}^{(1)}>0$
implies $\mathbf{T}_{ij}=0$ for $i<j$, so that from now on we can
indeed assume $\mathbf{T}_{ij}=0$ for all $i<j$.

Now we postmultiply \eqref{eq:ST} by \begin{equation}
\mathbf{L}=\begin{pmatrix}\mathbf{S}_{00} & 0 & \cdots & 0\\
\mathbf{S}_{10} & \mathbf{S}_{11} & \cdots & 0\\
\vdots & \vdots & \ddots & \vdots\\
\mathbf{S}_{l0} & \mathbf{S}_{l1} & \cdots & \mathbf{S}_{ll}\end{pmatrix}.\end{equation}
 Using Lemma~\ref{lem:DCMLOBCLem2} on the left-hand side and setting
$\boldsymbol{\Delta}=\mathbf{TL}=(\boldsymbol{\Delta}_{ij})$ on the
right-hand side we obtain \begin{equation}
\sum_{i=j}^{l}\epsilon_{i}^{(1)}\mathbf{S}_{li}\boldsymbol{\Delta}_{ij}=\left(\sum_{i=j}^{l}\epsilon_{i}^{(2)}\right)\mathbf{S}_{lj}\quad(0\leq j\leq l).\end{equation}
 Applying these matrix equations to the all-one column vectors $\mathbf{j}$
of the appropriate dimensions gives, in view of $\mathbf{S}_{lj}\mathbf{j}=\mathbf{j}$
and $(\mathbf{S}_{li}\boldsymbol{\Delta}_{ij})\mathbf{j}=\mathbf{S}_{li}(\boldsymbol{\Delta}_{ij}\mathbf{j})\leq\mathbf{S}_{li}\mathbf{j}=\mathbf{j}$,
the required inequalities $\sum_{i=j}^{l}\epsilon_{i}^{(2)}\leq\sum_{i=j}^{l}\epsilon_{i}^{(1)}$
($0\leq j\leq l$), which completes the proof of the theorem. \end{proof}

\subsection{Proof of Lemma\,\ref{lem:WCMLOBCLem1}\label{sub:ProfWCMLOBCLem1}}

\begin{proof} During the proof we write $\mathsf{Y}^{(i)}$, $i=1,2$,
for the subchannel outputs (here $\mathsf{Y}^{(i)}$ corresponds to
the probability vector $\epsilon^{(i)}$) and $\mathsf{Y}_{s}$, $s=0,1,2$,
for the dimension $s$ component of $\mathsf{Y}^{(i)}$ (corresponding
to the $s$-th block in the decomposition~\eqref{eq:WCMLOBExa}),
which is independent of $i$. We will use the (easily established)
fact that mutual information is linear in the following sense: \begin{align*}
I(\mathsf{X};\mathsf{Y}^{(1)}|\mathsf{U}) & =\sum_{s=0}^{2}\epsilon_{s}^{(1)}I(\mathsf{X};\mathsf{Y}_{s}|\mathsf{U}),\\
I(\mathsf{U};\mathsf{Y}^{(2)}) & =\sum_{s=0}^{2}\epsilon_{s}^{(2)}I(\mathsf{U};\mathsf{Y}_{s}),\end{align*}
 which generalizes to arbitrary decompositions of the form \eqref{eq:partition}.

Clearly $I(\mathsf{X};\mathsf{Y}_{0})=I(\mathsf{U};\mathsf{Y}_{0})=0$.
The (symmetric) channels $\mathsf{X}\to\mathsf{Y}_{1}$, $\mathsf{X}\to\mathsf{Y}_{2}$,
$\mathsf{U}\to\mathsf{Y}_{2}$ have channel matrices $\mathbf{S}_{21}$,
$\mathbf{I}_{7\times7}$, $\tfrac{\sigma}{6}\mathbf{J}_{7\times7}+\left(1-\tfrac{7\sigma}{6}\right)\mathbf{I}_{7\times7}$,
respectively. The channel $\mathsf{U}\to\mathsf{Y}_{1}$ has channel
matrix \[
\left(\tfrac{\sigma}{6}\mathbf{J}_{7\times7}+\left(1-\tfrac{7\sigma}{6}\right)\mathbf{I}_{7\times7}\right)\mathbf{S}_{21}=\tfrac{\sigma}{6}\mathbf{J}_{7\times7}+\left(1-\tfrac{7\sigma}{6}\right)\mathbf{S}_{21}\]
 The input distribution on $\mathfrak{U}$ (and hence the distribution
on $\mathfrak{X}$ as well) is uniform, this gives \begin{align*}
R_{2}(\sigma) & =I(\mathsf{U};\mathsf{Y}^{(2)})\\
 & =\epsilon_{1}^{(2)}\left(-H\left(\tfrac{2\sigma}{3}\right)+\log\tfrac{7}{3}-\tfrac{2\sigma}{3}\log\tfrac{4}{3}\right)\\
 & \quad+\epsilon_{2}^{(2)}\left(-H(\sigma)+\log7-\sigma\log6\right),\\
R_{1}(\sigma) & =I(\mathsf{X};\mathsf{Y}^{(1)}|\mathsf{U})\\
 & =\epsilon_{1}^{(1)}\left(H\left(\tfrac{2\sigma}{3}\right)+\tfrac{2\sigma}{3}\log\tfrac{4}{3}\right)\\
 & \quad+\epsilon_{2}^{(1)}\left(H(\sigma)+\sigma\log6\right),\end{align*}
 where $H(x)=-x\log x-(1-x)\log(1-x)$ denotes the binary entropy
function. To simplify the expressions below, we will take $\log$
as the natural logarithm, for which $H'(x)=\log\frac{1-x}{x}$, $H''(x)=-\frac{1}{x(1-x)}$.
We have further \begin{align*}
R_{2}'(\sigma) & =\epsilon_{1}^{(2)}\left(-\tfrac{2}{3}\log\tfrac{1-2\sigma/3}{2\sigma/3}-\tfrac{2}{3}\log\tfrac{4}{3}\right)\\
 & \quad+\epsilon_{2}^{(2)}\left(-\log\frac{1-\sigma}{\sigma}-\log6\right)\\
 & =-\epsilon_{1}^{(2)}\tfrac{2}{3}\log\tfrac{2(1-2\sigma/3)}{\sigma}-\epsilon_{2}^{(2)}\log\tfrac{6(1-\sigma)}{\sigma},\\
R_{1}'(\sigma) & =\epsilon_{1}^{(1)}\tfrac{2}{3}\log\tfrac{2(1-2\sigma/3)}{\sigma}+\epsilon_{2}^{(1)}\log\tfrac{6(1-\sigma)}{\sigma}.\end{align*}
 From this one verifies at once that $R_{1}'(\sigma)>0$, $R_{2}'(\sigma)<0$
for $0<\sigma<\frac{6}{7}$ (and $R_{1}'(\frac{6}{7})=R_{2}'(\frac{6}{7})=0$).
Hence, by results from standard calculus, $f$ is well-defined and
$f'\bigl(R_{1}(\sigma)\bigr)=\frac{R_{2}'(\sigma)}{R_{1}'(\sigma)}<0$,
so that $f$ is strictly decreasing. Moreover, since $R_{1}(0)=0$,
$R_{2}(0)=\epsilon_{1}^{(2)}\log\tfrac{7}{3}+\epsilon_{2}^{(2)}\log7=C_{2}$,
$R_{1}(\tfrac{6}{7})=\epsilon_{1}^{(1)}\log\tfrac{7}{3}+\epsilon_{2}^{(1)}\log7=C_{1}$,
$R_{2}(\tfrac{6}{7})=0$, we have $f\colon[0,C_{1}]\to[0,C_{2}]$,
$f(0)=C_{2}$, and $f(C_{1})=0$.

In order to decide whether $f$ is convex/concave/linear, we use the
second derivative test from standard calculus. We have to determine
the sign of \[
f''\bigl(R_{1}(\sigma)\bigr)=\frac{R_{2}''(\sigma)R_{1}'(\sigma)-R_{1}''(\sigma)R_{2}'(\sigma)}{R_{1}'(\sigma)^{3}}\]
 for $\sigma\in(0,\frac{6}{7})$, which is the same as the sign of
\begin{align}
 & R_{2}''(\sigma)R_{1}'(\sigma)-R_{1}''(\sigma)R_{2}'(\sigma)=\nonumber \\
 & =\frac{2(\epsilon_{1}^{(2)}\epsilon_{2}^{(1)}-\epsilon_{1}^{(1)}\epsilon_{2}^{(2)})}{\sigma(1-\sigma)(3-2\sigma)}\left((1-\sigma)\log\tfrac{6(1-\sigma)}{\sigma}\right.\nonumber \\
 & \left.-\left(1-\tfrac{2\sigma}{3}\right)\log\tfrac{2(1-\tfrac{2\sigma}{3})}{\sigma}\right).\end{align}
It may be verified that the right-hand factor \begin{align*}
g(\sigma) & =(1-\sigma)\log(1-\sigma)-\left(1-\tfrac{2\sigma}{3}\right)\log\left(1-\tfrac{2\sigma}{3}\right)\\
 & \quad+\tfrac{1}{3}\sigma\log\sigma+(1-\sigma)\log6-\left(1-\tfrac{2\sigma}{3}\right)\log2\end{align*}
 satisfies $g(0)=g(\frac{6}{7})=0$ and \[
g''(\sigma)=-\frac{1}{\sigma(1-\sigma)(3-2\sigma)}<0\quad\text{for \ensuremath{0<\sigma<\tfrac{6}{7}}},\]
 from which it follows that $g(\sigma)$ is positive in $(0,\frac{6}{7})$.
Hence the sign of $f''$ in $(0,\frac{6}{7})$ is constant and equal
to that of $\epsilon_{1}^{(2)}\epsilon_{2}^{(1)}-\epsilon_{1}^{(1)}\epsilon_{2}^{(2)}$.
This concludes the proof.\end{proof}

\subsection{Proof of Theorem\,\ref{thm:q_ary_Capacity_Thm}\label{sub:Profq_ary_Capacity_Thm}}

\begin{proof} It is clear from Lemma~\ref{lem:q_ary_Lem} that the
capacities of the subchannels are $C_{i}=(1-\rho_{i})\log|\mathfrak{X}|$
($i=1,2$). Further, for an arbitrary joint distribution $p(U,X)$
Lemma~\ref{lem:q_ary_Lem} gives \begin{align*}
C_{2}I(\mathsf{X}; & \mathsf{Y}_{1}|\mathsf{U})+C_{1}I(\mathsf{U};\mathsf{Y}_{2})\\
 & =(1-\rho_{2})\log|\mathfrak{X}|(1-\rho_{1})I(\mathsf{X};\mathsf{X}|\mathsf{U})\\
 & \quad+(1-\rho_{1})\log|\mathfrak{X}|(1-\rho_{2})I(\mathsf{U};\mathsf{X})\\
 & =\frac{C_{1}C_{2}}{\log|\mathfrak{X}|}\cdot I(\mathsf{X};\mathsf{X})\\
 & =\frac{C_{1}C_{2}}{\log|\mathfrak{X}|}\cdot H(\mathsf{X})\leq C_{1}C_{2},\end{align*}
which implies \eqref{eq:q_ary_CapacityRegion}.\end{proof}

\bibliographystyle{IEEEtran}
\bibliography{mathe,MyRef,StdRef,strings}
 
\end{document}